 \documentclass{nature}

\usepackage{amssymb}
\usepackage{amsmath,amsthm}
\usepackage{amsthm,thmtools}
\usepackage{amsfonts}
\usepackage{verbatim}

\usepackage{lineno}
\linenumbers

\usepackage{enumerate}

\usepackage[dvips]{epsfig}
\usepackage{color}

\theoremstyle{plain}
\newtheorem{theorem}{Theorem}
\newtheorem{proposition}{Proposition}

\theoremstyle{definition}
\newtheorem{definition}{Definition}

\newcommand{\LMD}{\lambda_0,\dots,\lambda_n}

\declaretheorem[name={Example},qed={\lower-0.3ex\hbox{$\square$}} ] {example}

\newtheorem{problem}{Problem}

\def\R{{\mathbb R}}

\def\W{{\mathbb W}}

\newcommand{\be}{\begin{equation}}
\newcommand{\ee}{\end{equation}}

\newcommand{\inter}{\operatorname{int}}
\newcommand{\updt}[1]{#1}  % use this to revert back to regular color.

\renewcommand{\vec}[1]{\mathnormal{#1}}

\newcommand{\scrv}{\emph{S.~cerevisiae}}

\makeatletter
\let\saved@includegraphics\includegraphics
\AtBeginDocument{\let\includegraphics\saved@includegraphics}
\renewenvironment*{figure}{\@float{figure}}{\end@float}
\makeatother

\nolinenumbers

\begin{document}
%%%%%%%%%%%%%%%%%5
%\frenchspacing

\title {Networks of ribosome flow models for modeling and analyzing intracellular traffic}

\author{Itzik Nanikashvili$^{1}$,Yoram Zarai$^{4}$, Alexander Ovseevich$^{2}$, Tamir Tuller $^{3,4}$, and Michael Margaliot$^{1,3}$}

\maketitle

\begin{affiliations}
 \item   School of Electrical Engineering, Tel-Aviv
University, Tel-Aviv 69978, Israel.
\item   Ishlinsky Institute for Problems in Mechanics,
Russian Academy of Sciences   
and the   Russian Quantum Center, Moscow, Russia.
 %%%
\item   Sagol School of Neuroscience, Tel-Aviv
University, Tel-Aviv 69978, Israel.
\item   Dept. of Biomedical Engineering, Tel-Aviv
University, Tel-Aviv 69978, Israel.
E-mail: tamirtul@post.tau.ac.il  \protect\\
\end{affiliations}

%\author{Itzik Nanikashvili, Yoram Zarai, Alexander Ovseevich, Tamir Tuller and
%Michael~Margaliot
 %\IEEEcompsocitemizethanks{
%\IEEEcompsocthanksitem
%I.  Nanikashvili is with the School of Electrical Engineering, Tel-Aviv University, Tel-Aviv 69978, Israel.
%E-mail: itzhakna@gmail.com
 %\IEEEcompsocthanksitem Y. Zarai is with the Dept. of Biomedical Engineering, Tel-Aviv University, Tel-Aviv 69978, Israel.
%E-mail: yoramzar@mail.tau.ac.il
%\IEEEcompsocthanksitem A. Ovseevich is with the Ishlinsky Institute for Problems in Mechanics,
%Russian Academy of Sciences, pr. Vernadskogo 101, 119526 Moscow, Russia. E-mail: ovseev@ipmnet.ru
%\IEEEcompsocthanksitem T. Tuller (Corresponding Author) is with the Dept. of Biomedical Engineering and the Sagol School of Neuroscience, Tel-Aviv University, Tel-Aviv 69978, Israel.
%E-mail: tamirtul@post.tau.ac.il
%\IEEEcompsocthanksitem
%M. Margaliot   is with the School of Electrical Engineering and the Sagol School of Neuroscience,
%Tel-Aviv University, Tel-Aviv 69978, Israel.
%E-mail: michaelm@eng.tau.ac.il  % <-this % stops a space
%}}

%\maketitle
%\begin{center}June  2, 2018\end{center}

\begin{abstract}
%%%%%%%%%%%%%%%%%
The ribosome flow model with input and output~(RFMIO) is a deterministic
dynamical  system that has been used to study    the flow of ribosomes during mRNA translation.
The  input of the~RFMIO controls its initiation rate and the output represents the
 ribosome exit rate (and thus the protein production rate) at the 3' end
of the mRNA molecule.
The~RFMIO and its variants  encapsulate  important properties that are relevant to modeling ribosome flow such as the possible evolution of ``traffic jams'' and non-homogeneous elongation  rates along the mRNA molecule,
and
 can also be used for studying additional intracellular processes such as transcription, transport, and more.

Here we consider networks of interconnected~RFMIOs as a fundamental tool for modeling, analyzing and re-engineering the complex mechanisms of protein production.
In these networks, the output of each RFMIO  may be divided, using connection weights,  between
 several inputs of other~RFMIOs. We show that under quite general feedback connections the network   has two important properties: (1)~it admits a unique steady-state and every
 trajectory converges to this steady-state; and
(2)~the problem of how to determine the connection weights so that
the network steady-state output is maximized is a convex optimization problem. These mathematical properties  make
 these networks highly suitable as models of various phenomena: property~(1) means that the behavior is predictable and ordered,
and property~(2) means that determining the optimal weights is
numerically tractable even for large-scale networks.

For  the specific case of a
 feed-forward network of RFMIOs we prove   an   additional  useful property, namely,
that there exists a     spectral representation  for the network
 steady-state, and thus it can be determined  without any numerical simulations of the dynamics.
We describe the implications of   these  results  to several fundamental biological phenomena and biotechnological objectives.
%%%%%
 \end{abstract}
%%%%%%%%%%%%%%%%%%%%%%%%%%%%%%%%%%%%%%%%%%%%%%%%%%%%%%%%%%%%%%%%%%%%%%%%%%%%

%\begin{IEEEkeywords}
%%
%Systems biology, biological networks,  mRNA translation, ribosome recycling, intracellular transport, synthetic biology, competition for shared resources, spectral, representation, eigenvalue sensitivity, convex optimization, maximizing protein production rate..
%%%
%\end{IEEEkeywords}

%\bigskip
%%%%%%%%%%%%%%%%%%%%%%%%%%%%%%%%%%%%%%%%%%%%
\section*{Introduction}\label{sec:intro}
%%%%%%%%%%%%%%%%%%%%%%%%%%%%%%%%%%%%%%%%%%%%

Gene expression is a complex multistage process in which information encoded in
the DNA is used to generate proteins or other gene products.
Gene expression involves two primary stages:  transcription and translation.
Each of these stages involves the sequential movement of enzymes along the genetic material.
During transcription, RNA copies of the DNA  genes are synthesized by enzymes called RNA~polymerase.
The product is the messenger RNA (mRNA), which codes, by a series of nucleotide triplets (called codons), the order in which amino-acids need to be combined to synthesize the protein.
%In prokaryotic protein-coding genes, the transcription creates mRNA that is ready for translation while in eukaryotic genes the product is a primary transcript of RNA (pre-mRNA), which first has to undergo some processing (e.g. RNA-splicing, etc) to become a mature mRNA.

Translation is the process in which the information in the mRNA is decoded and the protein is synthesized.
During translation, complex macromolecules called ribosomes  bind to the start codon in the mRNA and sequentially decode each codon to its corresponding  amino-acid  that is delivered to the awaiting ribosome by transfer RNA~(tRNA).
The amino-acid peptide is elongated until the ribosome reaches a stop codon,     detaches from the mRNA and the resulting amino-chain peptide is released, folded and becomes a functional protein~\cite{Alberts2002}.
The detached ribosome may re-initiate the same mRNA molecule (ribosome recycling~\cite{Skabkin2013,Steege1977}) or become available to translate other mRNAs.
To increase the translation efficiency, multiple ribosomes may decode the same mRNA molecule simultaneously (polysome)~\cite{Alberts2002}.

mRNA translation is a fundamental process  in all living cells of all organisms.
Thus, a better understanding of its bio-physical  properties  has numerous potential applications in many scientific disciplines including medicine, systems biology, biotechnology  and evolutionary biology.
Mechanistic models of translation are essential for:
 (1)~analyzing the flow of ribosomes along the mRNA molecule; (2)~integrating and understanding  the rapidly increasing experimental findings related to translation and its role
in the  dynamical  regulation of gene expression~(see, e.g., \cite{Dana2011,TullerGB2011,Tuller2007,Chu2012,Deneke2013,Racle2013,Zur2016,Dao2018,Lalanne2018}); and (3)~providing a computational testbed for predicting the effects of various
manipulations of the genetic machinery.
These models describe the dynamics of ribosome flow and include parameters
whose values   represent the various translation factors that affect the initiation rate and codon decoding times along the mRNA molecule.

Another fundamental biological process based on  the flow of biological ``machines''  along ``intracellular roads'' is intracellular transport.
In this process, vesicles are transferred to  particular intracellular locations by  molecular motors that haul them along microtubules and actin filaments~\cite{Alberts2002,Vale2003}.

We now review two computational models for ribosome
 flow that are the most relevant for this paper.
Numerous other models exist in the literature,
 see e.g. the survey papers~\cite{Zur2016,Haar2012}.

\subsection{Totally asymmetric simple exclusion process~(TASEP)}
%%%%%%%%%%%%%%%%%%%%%%%%%%%%%%%%%%%%%%%%%%%%%%%%%%%%%%%%%%%%%%%%%
 This is a  fundamental model in non-equilibrium statistical mechanics that has been extensively used to model and analyze translation and  intracellular transport.
 TASEP is a discrete-time, stochastic model describing particles hopping along an ordered  lattice of~$N$ sites~\cite{tasep0,macdonald1969concerning,TASEP_book}.
A particle at site $i$ may hop to site $i+1$ at a rate $\gamma_i$ but only if this site is empty. This models the fact that the particles have volume and thus cannot overtake one another. Specifically, a particle may hop to the first site at rate $\alpha$ (if the first site is empty), and hop out from the last site at rate~$\beta$ (if the last site is occupied). In the context of translation, the lattice of sites represents the chain of codons in the mRNA, and the hopping particles represent the moving ribosomes~\cite{TASEP_tutorial_2011,Shaw2003}.

Analysis of TASEP is in general non trivial, and closed-form results have been obtained mainly for the homogeneous TASEP, i.e. the case where all $\gamma_i$s are assumed to be equal. The non-homogeneous case is typically studied  via extensive and time-consuming Monte Carlo simulations.

In each cell, multiple translation processes take place concurrently, utilizing the limited shared translation resources (i.e. ribosomes and translation factors). For example, a yeast cell
contains about~$60,000$ mRNA molecules and about~$240,000$ ribosomes~\cite{Haar2008,Zenklusen2008}.
The competition for shared resources induces indirect  interactions and correlations between the various translation processes.
Such interactions must be considered when analyzing the cellular economy of the cell, and also when designing synthetic circuits~\cite{arkin_context_2012,DELVECCHIO2018}.

Analyzing large-scale translation, as opposed to translation of a single isolated mRNA molecule, is thus an important research direction that is recently attracting considerable  research attention~\cite{Raveh2016,mixed_tasep,brackley2011dynamics,gyorgy2015isocost,mather2013translational,ZaraiTuller2018}. For example, in~\cite{mather2013translational}, a TASEP-based computational network consisting of $400$ mRNA species and $14,000$ ribosomes has been
 used to analyze  the sensitivity of a translation network to perturbations in the initiation and elongation rates and in the mRNA levels. In \cite{ZaraiTuller2018}, a deterministic mean-field approximation of  TASEP, called  the \emph{ribosome flow model}~(RFM), was used for studying the effect of fluctuations in
the mRNA levels  on translation in a whole cell simulation of  an \updt{{\scrv}} cell.

Here, we consider   large-scale networks of interconnected
translation processes, whose building blocks are~RFMs with suitable inputs and outputs.

%%%%%%%%%%%%%%%%%%%%%%%%%%%%%%%%%%%%%%%%%%%%
\subsection{Ribosome flow model (RFM)}
%%%%%%%%%%%%%%%%%%%%%%%%%%%%%%%%%%%%%%%%%%%%

The RFM~\cite{rfm0} is a continuous-time, deterministic model for ribosome flow  that can be obtained via a mean-field approximation of TASEP with open boundary conditions (i.e., the two sides of the TASEP lattice are connected to two particle reservoirs)~\cite{Zarai20170128}.

In the RFM, the mRNA molecule is coarse-grained into a chain of $n$ consecutive sites of codons. For each site  $i\in \{1,2,\dots,n\}$, a state-variable $x_i(t)\in [0,1]$ denotes the normalized ribosomal occupancy level (or ribosomal density) at site $i$ in time $t$, where $x_i(t)=1$ $[x_i(t)=0]$ means that site~$i$ is completely full [empty] at time~$t$.
The RFM is characterized by $n+1$ positive parameters: the initiation rate ($\lambda_0$), the transition rate from site $i$ to site $i+1$ ($\lambda_i$), and the exit rate ($\lambda_n$). Ribosomes that attempt to bind to the first site are constrained by the occupancy level of that site, i.e. the effective flow of ribosomes into the first site is given by~$\lambda_0(1-x_1)$.
This means that: (1) the maximal possible     entry rate is~$\lambda_0$; and
(2) the entry rate decreases as the first site becomes fuller, and becomes zero when the first site is completely full. Similarly, the effective flow of ribosomes from site~$i$ to site~$i+1$ increases [decreases] with the occupancy level at site~$i$ [$i+1$] and thus is given by~$\lambda_i x_i (1-x_{i+1})$. This is a ``soft'' version of the~\emph{simple exclusion} principle that models the fact that ribosomes have volume and cannot overtake one another, thus as the occupancy level at site~$i+1$ increases less ribosomes can enter this site, and the effective flow rate from site~$i$ to site~$i+1$ decreases.
The (soft) simple exclusion principle in the~RFM allows to model the evolution of ribosomal ``traffic jams''. Indeed, if site~$i$ becomes fuller, i.e.~$x_i$ increases then the flow from site~$i-1$ to site~$i$ decreases and thus site~$i-1$ also becomes fuller, and so on.
 Recent findings suggest that in
 many organisms and conditions a non-negligible percentage of
 ribosomes tends to be involved in such traffic jams (see, for example,~\cite{Diament2018}).

The dynamics of the RFM with $n$ sites are given by $n$ nonlinear first-order ordinary differential equations~(ODEs) describing the change in the occupancy level of each site as a function of time:
\begin{align}\label{eq:rfm}
%%%
                    \dot{x}_1 (t)&=\lambda_0 (1-x_1(t)) -\lambda_1 x_1(t)(1-x_2(t)), \nonumber \\
                    \dot{x}_2(t)&=\lambda_{1} x_{1}(t) (1-x_{2}(t)) -\lambda_{2} x_{2}(t) (1-x_3(t)) , \nonumber \\
                    \dot{x}_3 (t)&=\lambda_{2} x_{ 2}(t) (1-x_{3}(t)) -\lambda_{3} x_{3}(t) (1-x_4(t)) , \nonumber \\
                             &\vdots \nonumber \\
                    \dot{x}_{n-1}(t)&=\lambda_{n-2} x_{n-2}(t) (1-x_{n-1}(t)) -\lambda_{n-1} x_{n-1}(t) (1-x_n(t)),
										\nonumber \\
                    \dot{x}_n(t)&=\lambda_{n-1}x_{n-1}(t) (1-x_n(t)) -\lambda_n x_n(t).
\end{align}
 Note that the~$x_i$s are dimensionless,
and the~$\lambda_i$s have units of~$1/\text{time}$.
The protein~\emph{production rate} or~\emph{translation rate}  is the rate in which ribosomes detach from
 the mRNA at time $t$, that is, $r(t):=\lambda_n x_n(t)$.

If we let $x_0(t):=1$ and $x_{n+1}(t):=0$, then~\eqref{eq:rfm} can be written more succinctly as
\be\label{eq:rfm_all}
\dot{x}_i=h_{i-1}(x)-h_i(x),\quad i=1,\dots,n,
\ee
where~$h_i(x):=\lambda_i x_i(1-x_{i+1})$,
and  we omit the dependence in time for clarity.
 This means that the change in the density at site~$i$ is the flow from site~$i-1$ to site~$i$ minus the flow from site~$i$
to site~$i+1$.

Let $\vec{x}(t,\vec{a})$ denote the solution of~\eqref{eq:rfm} at time $t\ge 0$
 for the initial condition $\vec x(0)=\vec a$.
Since the state variables correspond to normalized occupancy levels, we always assume that $\vec a$ belongs to the closed $n$-dimensional unit  cube denoted~$[0,1]^n$. Let~$(0,1)^n$ denote the interior of~$[0,1]^n$. In other words, $a\in(0,1)^n$ means that every entry~$a_i$  of~$a$
satisfies~$0<a_i<1$.

It was shown in~\cite{RFM_stability} that~$[0,1]^n$ is an invariant set of the dynamics i.e. if~$\vec a\in [0,1]^n$ then~$\vec x(t,a) \in[0,1]^n$ for all~$t\geq0$.
It was also shown that the RFM is a~\emph{tridiagonal cooperative dynamical system}~\cite{hlsmith,fulltppaper},
and that this implies that~\eqref{eq:rfm} admits a  steady-state
  point~$\vec e=\vec e(\LMD) \in (0,1)^n$, that is globally asymptotically stable, that is,
	\[
	\lim_{t\to\infty} \vec x(t,a)=\vec e, 
	\text{ for any } \vec a\in [0,1]^n 
	\]
	 (see also~\cite{RFM_entrain}). In particular, the production rate converges to the steady-state value
	\[
	r_{ss}:=\lim_{t\to\infty} r(t) =\lambda_n e_n.
\]
This means that the parameters of the RFM determine a unique
steady-state occupancy at all sites along the mRNA. At the steady-state the flow into every site is equal to the flow out of the site.
For any initial density the dynamics converge to this steady-state.

\subsubsection*{Spectral representation of the RFM steady-state}
%%%%%%%%%%%%%%%%%%%%%%%%%%%%%%%%%%%%%%%%%%%%%%%
Simulating the dynamics of  large-scale RFMs until (numerical)
convergence to the steady-state may be tedious.
A useful property of the RFM is that the steady-state can be computed using a spectral approach, that is, based on calculating the eigenvalues and eigenvectors of a suitable matrix~\cite{rfm_max}.
Consider the RFM with dimension $n$ and rates $\LMD$. Define the $(n+2)\times(n+2)$ Jacobi  matrix
\be\label{eq:bmatrox}
                A(\lambda_0,\dots,\lambda_n):= \begin{bmatrix}
             %     J:= \left[ \begin{smallmatrix}
%%
 0 &  \lambda_0^{-1/2}   & 0 &0 & \dots &0&0 \\
\lambda_0^{-1/2} & 0  & \lambda_1^{-1/2}   & 0  & \dots &0&0 \\
 0& \lambda_1^{-1/2} & 0 &  \lambda_2^{-1/2}    & \dots &0&0 \\
%%%
 & &&\vdots \\
%%%
 0& 0 & 0 & \dots &\lambda_{n-1}^{-1/2}  & 0& \lambda_{n }^{-1/2}     \\
%%%
 0& 0 & 0 & \dots &0 & \lambda_{n }^{-1/2}  & 0
%%%
 \end{bmatrix}.
 % \end{smallmatrix}\right].
\ee
This is a   symmetric matrix, so
   its eigenvalues are real. Since~$A$ is componentwise non-negative  and irreducible,  it admits a unique maximal eigenvalue~$\sigma>0$ (called the Perron eigenvalue or Perron root),
 and the corresponding eigenvector~$\zeta\in\R^{n+2}$  (the Perron eigenvector)
has positive entries~\cite{matrx_ana}.

\begin{theorem}\label{thm:spect} \cite{rfm_max}
%%%%%%%%%%%%%%%%%%%%%%%%%%%%%%%%%%%%%%%%%%%%%%%%%%%%
Consider an RFM with dimension $n$ and rates $\LMD$.
Let~$A$ be the matrix defined in~\eqref{eq:bmatrox}. Then the steady-state values of the~RFM satisfy:
\begin{align}\label{eq:spect_rep}
r _{ss}=\sigma^{-2}  \text{ and }
e_i =\lambda_i^{-1/2}\sigma^{-1}\frac{\zeta_{i+2}}{\zeta_{i+1}}, \quad i=1,\dots,n.
\end{align}
\end{theorem}
%%%%
In other words,   the steady-state density and production rate in the RFM can be
obtained from the Perron eigenvalue and eigenvector  of~$A$. In particular, this makes it possible to
determine~$r_{ss}$ and~$e$ even for very large chains
using     efficient and numerically stable algorithms for computing
 the eigenvalues and eigenvectors of a Jacobi matrix (see, e.g.,~\cite{fernando_97}).

Consider two RFMs: one with rates~$\lambda_0,\dots,\lambda_n$
and the second with    rates~$\tilde \lambda_0,\dots,\tilde   \lambda_n$ such that~$\tilde \lambda_i=\lambda_i$ for all~$i$ except for one index~$k$ for which~$\tilde \lambda_k>\lambda_k$.  Then~$\tilde \lambda_k^{-1/2} <\lambda_k^{-1/2} $.
Let~$\sigma$ [$\tilde \sigma$] denote the Perron root of the matrix~$A:=A(\lambda_0,\dots,\lambda_n)$ [$\tilde A:=A(\tilde \lambda_0,\dots,\tilde \lambda_n)$].
Comparing  the entries of~$A$ and~$\tilde A$, it follows from
known results in
the  Perron-Frobenius theory that~$\tilde \sigma< \sigma$.
Hence,~$\tilde \sigma^{-2} >   \sigma^{-2}$,
so Thm.~\ref{thm:spect} implies that
the steady-state production rate in the~RFM
increases when one \updt{(or more)} of the rates increases.

Thm.~\ref{thm:spect} has several more important implications. For example, it implies
that~$r_{ss}=r_{ss}(\LMD)$ is a \emph{strictly concave function} on $\R^{n+1}_{++}$~\cite{rfm_max}.
Also, it implies that the sensitivity of  the steady-state with respect to~(w.r.t.)
 a perturbation in the translation
rates
becomes an eigenvalue sensitivity problem~\cite{RFM_sense}. We refer to the
survey~\cite{rfm_chap} for more details.

Ref.~\cite{mt} extended the RFM into a single-input single-output (SISO) control system, by defining the production  rate as an output, and by introducing a time-varying input   $u:\R_+ \to \R_+$ representing the flow of ribosomes from the ``outside world''
 into the mRNA molecule. This is referred to as the~\emph{RFM with input and output}~(RFMIO).  The RFMIO dynamics is  thus described by:
\begin{align*}
%%%
\dot{x}_1&=u\lambda_0(1-x_1)-\lambda_1x_1(1-x_2), \\
\dot{x}_2&=\lambda_1x_1(1-x_2)-\lambda_2x_2(1-x_3), \\
&\vdots \\
\dot{x}_n&=\lambda_{n-1}x_{n-1}(1-x_n)-\lambda_nx_n,\\
y&= \lambda_nx_n.
\end{align*}
%%%%%%
Note that~$u$ multiplies~$\lambda_0$, so that the initiation rate at time~$t$ is~$u(t)\lambda_0$.
We consider~$\lambda_0$ as modeling an intrinsic bio-physical property of the mRNA,  and~$u(t)$
as an ``outside'' effect e.g. the time-varying abundance
 of ``free'' ribosomes in the vicinity of the~mRNA
(see Fig.~\ref{fig:rfm_strand_n}).
 Throughout, 
we always assume that~$u(t)>0$ for all~$t\geq 0$ in order to avoid some technical problems arising when the initiation rate is zero.
Of course, for $u(t)\equiv c$, with~$c>0$,
the~RFMIO  becomes an~RFM  with initiation rate~$c\lambda_0$.
In this case, the convergence to steady-state represents a form of homeostasis that is sensitive to the value of the input.
In particular, the steady-state density of such an RFMIO can be computed using the spectral approach described
in Thm.~\ref{thm:spect}.

\begin{figure}
	\centering
	%\scalebox{0.8}{\input{fig_rfm_strand_n.eps_tex}}
	 \includegraphics[scale=0.7]{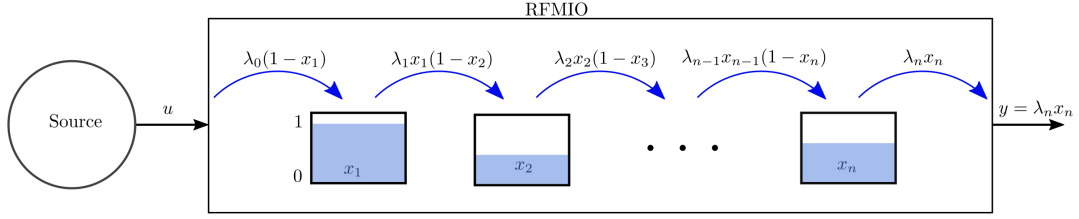}
	\caption{An RFMIO of length $n$,
	output~$y$, and input~$u $   from  an external source. \label{fig:rfm_strand_n}}
\end{figure}

We write the  RFMIO dynamics more succinctly  as
%%%%%%%%
\begin{align}\label{eq:rfmio}
\dot{\vec x}&=\vec f(\vec x,u), \nonumber \\
y&= \lambda_n x_n.
\end{align}
  Let $x(t,a,u)$ denote the solution of the RFMIO at time~$t$ given the initial condition $x(0)=a$ and input~$u$.

The RFMIO   facilitates modeling a network of interconnected ``roads'' (e.g.  mRNA or DNA molecules, microtubules, etc.), where the flow out  of one intracellular road (representing   ribosomes, RNAPs, vesicles attached to molecular motors, etc.) may enter another  road  in the network or re-enter the same  road. This enables the analysis of important phenomena such as translation re-initiation \cite{Kochetov2008}, competition for finite resources  including the   effect of the exit rate from one road
on  the initiation rate of other intracellular  roads (see e.g. \cite{Francesca2015}), transport on a network of interconnected microtubules \cite{Vale2003}, etc.
Such interconnected models are essential for engineering cells for various biotechnological objectives such as optimization of protein production rate, optimization of growth rate, and optimization of traffic jams, as the multiple processes
taking place in the cell  cannot be analyzed using models of a single ``road" \cite{Glick1995}.

In the context of translation,
the output of an RFMIO represents both the flow
of ribosomes out of the   mRNA molecule and the
 synthesis rate of proteins. If the output of a RFMIO is divided into several inputs of other  RFMIOs then this represents the distribution of the exiting  ribosomes   initiating \updt{the} other mRNAs.

%We believe that such networks will play an important role in modeling and analyzing mRNA translation in the cell (rather than in a single, isolated mRNA molecule).

In this paper, we consider networks of interconnected RFMIOs.
We show that such dynamical networks
provide a useful and versatile modeling tool for
many  dynamical intracellular traffic phenomena.
Our first
main result shows that under quite general feedback connections the network
admits a unique steady-state and every solution of the dynamics converses to this steady-state. 
In other words, the network is
 \emph{globally asymptotically  stable}~(GAS).
This is important for several reasons. For example,
 GAS implies that the network admits an ordered behavior and
   paves the way for analyzing further important questions e.g. how does the steady-state \updt{depends} on various parameters?

We  analyze the problem of maximizing the steady-state output of the network.  Specifically, the question we consider  is to determine the interconnection weights values between the RFMIOs in the network so that the network output is maximized. Our second main result
 shows that this is a convex optimization problem, implying that  it can be solved using highly efficient algorithms even for very large networks.

In the specific case of feed-forward networks of RFMIOs, we
show an additional property, namely,
 that we can determine the steady-state of the entire network using a spectral approach, and with no need to numerically solve the dynamical equations.

We note that two previous papers considered specific networks of RFMs. Ref~\cite{mt} studied the effect of ribosome recycling
using a single~RFMIO with positive feedback from the
output (i.e. production rate)
to the input (i.e. initiation rate). Ref.~\cite{Raveh2016} analyzed a closed system composed of a dynamic free pool of  ribosomes
that feeds a single-layer of parallel~RFMIOs. This was used as a tool for analyzing the indirect effect between different mRNA molecules due to
competition for a shared resource, namely, the pool of free ribosomes.
%%%In the current study we analyze a model which is a generalization of the two previous models.
Here, we study  networks that provide a significant generalization of these particular models.

The remainder of this paper is organized as follows. The next section describes the network of RFMIOs that we introduce and analyze in this paper. \updt{Then we present our main results and demonstrate them using a biological example}. The final section summarizes and describes several directions for future research. To increase the readability of this paper, all the proofs are placed in the Appendix.

%Section~\ref{sec:Analysis-of-ff-networks}

We use standard notation. Vectors [matrices] are denoted by small [capital] letters.
$\R^n$ is the set of vectors with~$n$ real coordinates. $\R^n_+$ [$\R^n_{++}$] is the the set of vectors with~$n$ real and nonnegative [positive] coordinates.
For a (column) vector $\vec x\in\R^n$, $x_i$ is the $i$-th entry of $\vec x$, and~$\vec x'$ is the transpose
 of~$\vec x$.
If a time-dependent variable~$x(t)$ admits a steady-state then we
denote it by~$x_{ss}$,
 that is,~$x_{ss}:=\lim_{t\to\infty}x(t)$.

%%%%%%%%%%%%%%%%%%%%%%%%%%%%%%%%%%%%%%%%%%%%
\section*{Networks of RFMIOs} \label{sec:net_rfmio}
%%%%%%%%%%%%%%%%%%%%%%%%%%%%%%%%%%%%%%%%%%%%
Consider a  network of $m$ interconnected RFMIOs.
The \emph{input} to the network is   a source whose  output rate is~$y_0$, and
 represents external resources that drive the elements in the network.
For example, this can represent  pools of free ribosomes in the cell.
%The source  feeds some (at least one) of the RFMIOs in the network.
The \emph{output} of the entire network is denoted by  $y$.
This may represent for example the flow of a desired  protein produced by the network, the total flow of ribosomes that
feed some other process, etc.

%Each RFMIO is a node in the network and the edges represent the weighted interconnections between the RFMIOs.
%A proportion of the  output~$y_\alpha$ of node~$\alpha$ may be connected to the input~$u_\beta$ of
%node~$\beta$  with a  proportion parameter (or edge weight)
% $v_{\alpha\beta} \in [0,1]$. This  means that~$u_\beta$ includes the term~$v_{\alpha\beta}  y_\alpha$.
%Thus, the total input of  node~$\beta$ is $u_\beta= \sum_\gamma v_{\gamma\beta} y_\gamma $, where the
%summation is over all nodes~$\gamma$ such that (part of) the output~$y_\gamma$ is connected to~$\beta$.

For $i\in \{1,...,m\}$, RFMIO $i$ is a dynamical  system  with dimension~$n^i$,
input~$u^i$ and output~$y^i$.
% The portion of the source output  connected to $u_j$ is denoted by $v_{0,j}$, this means that~$u^j$ includes the term~$v_{0,j}  y^0$ .
For any~$k\in \{0,1,...,m\}$ some or all of the output~$y^k$ may be connected to the input of another~RFMIO, say, RFMIO~$j$
 with a control  parameter (or  weight)~$v_{k,j} \in [0,1]$. Here~$v_{k,j}=0$ means that $y^k$ is not connected to $u^j$. The input to RFMIO~$j$
is thus~$u^j=\sum_{k=0}^{m} v_{k,j}  y^k$.

We define the total   network output $y=y^{m+1}$   by
\[
y(t):=\sum_{j=1}^{m} v_{j,m+1} y^j(t),
\]
where  $v_{j,m+1}\in [0,1]$ are the proportion weights from the output $y^j$
of RFMIO~$j$  to the network output~$y$.

We say that the network is \emph{feasible} if:
 (1)~every~$v_{k,j} \in [0,1]$; and (2)~$\sum_{j=1}^{m+1} v_{k,j}=1$.
The first requirement corresponds to the fact that every~$v_{k,j}$ describes the \emph{proportion} of the output~$y_k$ that feeds the input of RFMIO~$j$.
The second requirement means that~$ \sum_{j} v_{k,j} y_k=y_	k$ i.e.
  the connections indeed describe a distribution
	of the output~$y_k$ to other points in the network.

%%%
%We say that the network is \emph{feasible} if (1)~every~$v_{\alpha\beta} \in [0,1]$;
%and (2)~$\sum_{\beta} v_{\alpha\beta}=1$. The first requirement corresponds to
% the fact that every~$v_{\alpha\beta}$ describes the \emph{proportion}
%of the output~$y_\alpha$ that feeds the input of RFMIO~$\beta$. The second requirement
%means that~$ \sum_{\beta} v_{\alpha\beta} y_\alpha=y_\alpha$ i.e.
%when we sum up all the divisions of the output we obtain the total output.

%The network may also include fixed sources
%connected to some RFMIO inputs. These
%represent external resources that   feed some of the
% RFMIOs in the network. For example, this can represent
%  free pools of ribosomes in the cell.

  We use  $x_i^j(t)$, $i=1,\dots,n^j$,  to
denote the state-variable describing the occupancy at site~$i$ in RFMIO~$j$ at time~$t$.
The vector
\[
z(t):=\begin{bmatrix}
 x_1^1(t),\cdots, x_{n^1}^1(t), x_1^2(t),\cdots, x_{n_2}^2(t),\cdots,
x_1^m(t), \cdots,x_{n^m}^m(t) \end{bmatrix}' \in [0,1]^\ell
\]
aggregates   all the state-variables in the network, where $\ell:=\sum_{j=1}^m n^j$.
The  variables
\[
v:=\{v_{k,j}\}, \quad k\in \{0,1,\dots,m\},  \quad j\in \{1,\dots,m+1\}
\]
 describe the connections between the RFMIOs in the network.
%The  \emph{output} of the network at time $t$ is
%\[
%y(t):=\sum_{\gamma\in \mathcal O} w_\gamma y_\gamma,
%\]
%where~$w_\gamma $ are positive weights, and~$\mathcal O \subset \{1,\dots,m\}$.

We demonstrate using several examples how networks of RFMIOs, with and without feedback connections,
 can be used to model and study
various intracellular networks.  As we will see below, our main theoretical result
guarantees that all these networks are~GAS. Thus, the state-variables in the networks converge to a steady-state that depends on the various parameters, but not on the initial condition.

Our first example describes
the efficiency of ribosome recycling in eukaryotic mRNA, and the tradeoff
between recycling on the one-hand and the need to ``free'' ribosomes for other mRNAs on the other-hand.

%%%%%%%%%
 \begin{example} \label{ex:feedback}
%%%%
	Consider the system depicted in Fig.~\ref{fig:ex_feedback}.
	Here a fixed  source with rate~$0.1$
	is feeding an~RFMIO  of length $n=3$ with rates $\lambda_0=\dots=\lambda_3=2$. A proportion  $v\in[0,1]$
	of the RFMIO output~$y(t)$ is fed back into the input,
	so that the total RFMIO input is~$u(t)=0.1+vy(t)$.
	Fig.~\ref{fig:ex_feedback_res}
	shows the steady-state
	values~$y_{ss}(v)$ and~$(1-v)y_{ss}(v)$ as a function of $v$.
	Note that for
	both these   functions there is a unique maximizing  value of~$v$.
	%%%%%%
We may interpret~$y$ as the total production rate, and~$(1-v)y$
	as the rate of ribosomes that are \emph{not} recycled,  and thus can be used to translate other mRNA molecules. Of course, one can also define other
	functions as the network output, say, some    weighted sum of the production rate and the
	rate of non-recycled ribosomes. In this case, finding the
	value~$v$ that maximizes the steady-state  output
	corresponds to maximizing the
	production rate on a specific mRNA molecule
	while still ``freeing'' enough 
	ribosomes for other purposes.
	%%%%
\end{example}

\begin{figure}
	\centering
	%\scalebox{0.8}{\input{fig_ex_feedback.eps_tex}}
	 \includegraphics[scale=0.9]{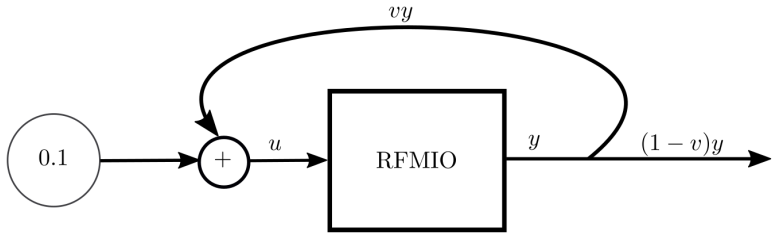}
	\caption{The network in Example~\ref{ex:feedback}. The RFMIO contains $n=3$ sites with rates~$\lambda_0=\dots=\lambda_3=2$.}
	\label{fig:ex_feedback}
\end{figure}

\begin{figure}
 \begin{center}
%%%%%%%%%%%%%%%%%%%%%%%[width=12cm,height=12cm]
  \includegraphics[scale=0.6]{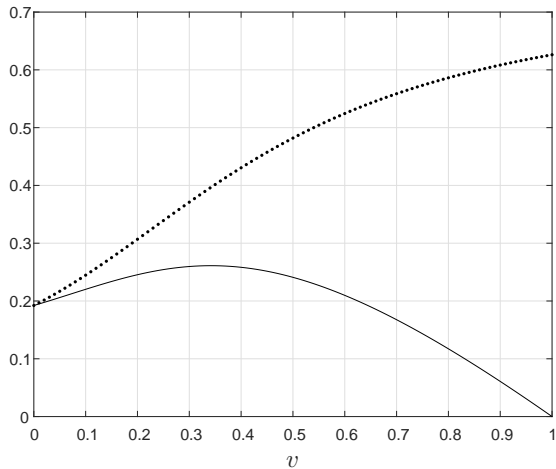}
	\caption{Steady-state values~$y_{ss}(v)$ (dotted line) and~$(1-v)y_{ss}(v)$ (solid) for the system  in
	 Example~\ref{ex:feedback}    as a function of $v$.}
	\label{fig:ex_feedback_res}
%%%%%%%%%%%%%%%%%%%%%%%%%%
\end{center}
\end{figure}

Our first result applies to quite general networks of interconnected~RFMIOs.
%%%%%%%%%%%
\begin{theorem}\label{thm:ss}
%%%%%%%%%%%%%%%%%%%%%%
A feasible   network of $m$ RFMIOs admits a
globally asymptotically stable steady-state point~$e\in (0,1)^\ell$, i.e.
\[
\lim_{t\to\infty}z(t,a)=e, \text{ for all } a\in [0,1]^\ell.
\]
\end{theorem}

Theorem~\ref{thm:ss} implies that all the RFMIO 
 state-variables (and thus the network output)
 converge to a unique steady-state value.
We assume throughout that at steady-state every~$u^i_{ss}$
is positive. Indeed,~$u^i_{ss}=0$ implies that~$y^i_{ss}=0$
and thus RFMIO~$i$ can simply be deleted from the network.

\begin{example} \label{ex:rfms_serial}
%%%%%%%%%%%%%%%%%%%%%%%%%%%%%%%%%%%%%%%%%%%%
Consider a network with $m=2$ RFMIOs, where RFMIO~1 is fed with a unit source, and the output of RFMIO~1 feeds the input of RFMIO~2 (see
 Fig.~\ref{fig:serial_2rfmio}).
The output of RFMIO~2 is defined as the network output~$y$. Both RFMIOs have dimension~$n=3$, RFMIO~1 with rates $[\lambda_0^1, \lambda_1^1, \lambda_2^1, \lambda_3^1] = [1, 1, 1/4, 1]$, and RFMIO~2 with rates~$[1,1,1,1]$.
Fig.~\ref{fig:ex_rfms_serial}
depicts the state-variables~$x^1_i(t)$ of RFMIO~1 and~$x^2_i(t)$ of RFMIO~2, $i=1,2,3$, as a function of~$t$, for the initial condition~$x^1_i(0)=x^2_i(0)=1/10$.
Note that the rate~$\lambda_2=1/4$ in RFMIO~1 leads to a ``traffic jam'' of ribosomes in this RFMIO, that is, the steady-state densities in the first two sites are high, \updt{whereas the density in the third site is  low}.
This yields
a low output rate from this RFMIO. The second RFMIO thus converges to a steady-state with low densities.
%%%%%%%%%%%%
\end{example}

\begin{figure}
	\centering
	%\scalebox{1.25}{\input{serial_2rfmio.eps_tex}}
	 \includegraphics[scale=1.2]{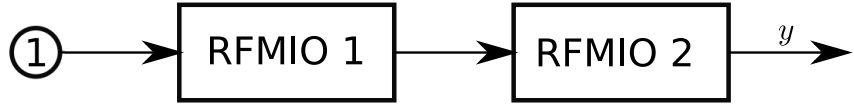}
	\caption{Network of two serially connected
	RFMIOs    in Example~\ref{ex:rfms_serial}.} \label{fig:serial_2rfmio}
\end{figure}

\begin{figure}
 \begin{center}
\includegraphics[width= 9cm,height=7cm]{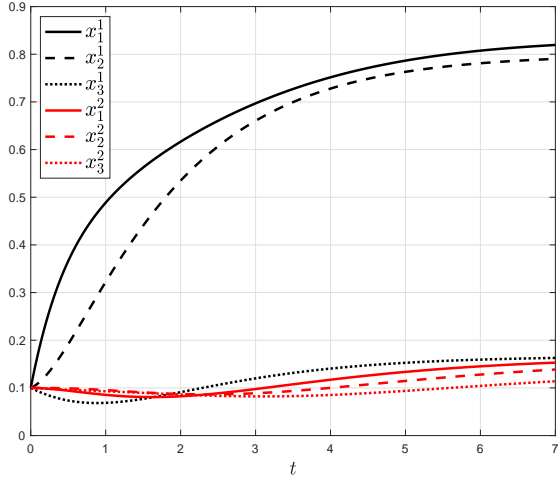}
\caption{The state variables $x^1_i(t)$ of RFMIO~1 and $x^2_i(t)$ of RFMIO~2, $i=1,2,3$, as a function of $t$ for the network in Example~\ref{ex:rfms_serial}.}\label{fig:ex_rfms_serial}
 \end{center}
\end{figure}

Example~\ref{ex:rfms_serial} may represents a case of re-initiation: one ORF appears in the 5'UTR of the second ORF and the ribosomes finishing the translation of the first ORF start translating the second one \cite{Luukkonen1995,Kozak2001,Kochetov2008}. In this case, a low elongation rate along the
first ORF is expected to yield a low density of ribosomes in the second ORF.

\subsection{Optimizing the Network Output Rate}
%%%%%%%%%%%%%%%%%%%%%%%%%%%%%
Several papers considered optimizing the production rate in  a \emph{single}~RFM or some variant of the~RFM~\cite{alexander2017,rfm_down_regul,rfm_max,HRFM_concave,RFM_r_max_density}.
Here, we study a different problem, namely, maximizing
 the steady-state output in a network of RFMIOs w.r.t. the control (or connection) weights
in the network. In other words, the problem is how to
distribute  the traffic between the different RFMIOs  in the network so that
the steady-state output is maximized.
\begin{problem}\label{prob:maxoutgen}
%%%%%%%%%%%%%%%%%%%%%%%%%%%%%
Given  a   network of $m$ RFMIOs with a   network output
\[
y(t):=\sum_{j=1}^{m} v_{j,m+1} y^j(t),
\]
maximize the steady-state network output~$y_{ss}:=\lim_{t\to\infty} y(t)$
 w.r.t. the control variables $v_{k,j}$ subject to the constraints:
	\begin{align}\label{eq:constrain}
	%%%%%%%%
	&v_{k,j} \in[0,1] \text{ for all }   k\in \{0,1,\dots,m\},  \; j\in \{1,\dots,m+1\} , \nonumber \\
	& \sum_{j} v_{k,j}=1 \text{ for all }   k\in \{0,1,\dots,m\}.
	\end{align}
	%%%%
\end{problem}

The constraints here guarantee the feasibly of the network.
However, the results below remain valid even if the second constraint in~\eqref{eq:constrain} is replaced by~$\sum_{j} v_{k,j}\leq 1 $
 for all~$k\in \{0,1,\dots,m\}$.

 The next result is instrumental for
 analyzing Problem~\ref{prob:maxoutgen}.

\begin{proposition}\label{prop:concavegen}
%%%%%%%%%%%%%%%%%%%%%%%%%%
%Consider a     network of   RFMIOs.
% Let $v$ denote the collection of all control weights.
%The mapping $v \to y_{ss}(v)$ is not necessarily   concave. 
Under a suitable reparametrization
	Problem~\ref{prob:maxoutgen} becomes a convex optimization problem.
%%%%%%%%
\end{proposition}
%%%
The following example demonstrates this result.
 %%%%%%%%%%%%%%%%%%%%%%%%%%%%%%%%%%%%%%%%%%%%%%%%%%%%%%%%%%
\begin{example}\label{exa:noncon}
%%%
Consider an RFMIO with a single site  and rates~$\lambda_0=\lambda_1=1$:
\begin{align}\label{eq:scarfm}
						\dot x_1&= (1-x_1)u -x_1,\nonumber\\
						y& = x_1.
\end{align}
Suppose that the input is~$u=0.1+v y$, where~$v\in[0,1]$, that is, there is a feedback connection
from the output of the~RFMIO back to the input with a weight~$v$.
It is straightforward to verify that for any~$x_1(0) \in [0,1]$ the solution~$x_1(t)$   converges to
the value
\[
			e_1(v) := \begin{cases} \frac{  v-1.1+\sqrt{  (1.1-v)^2 +0.4 v}    }{2v} & \text{if } v>0,\\
																1/11 & \text{if } v=0.\end{cases}
\]
%%%
It is also straightforward to verify that~$\frac{d^2}{dv^2}e_1>0$ for all~$v\in (0,1)$,
so~$e_1$ and thus the steady-state output~$y_{ss}(v)=e_1(v)$
is \emph{not} concave in~$v$.
We conclude that  the optimization problem:
\be\label{prob:opii}
\max y_{ss}(v) \text{ subject to } v\in [0,1]
\ee
 is \emph{not}
 a convex optimization problem.
 Nevertheless, since in this particular case~$y_{ss}(v)$ is a scalar function,
 it is easy to solve
this optimization problem yielding (all numerical values in this paper are to four digit accuracy)
\be\label{eq:yssopt}
y_{ss}^*:=y_{ss}(1)=0.2702.
\ee

The  reparametrization is based on
redefining the input  as~$u=0.1+w$, with the constraint~$w\in[0, y]$.
 Now the steady-state output of~\eqref{eq:scarfm}  is
$y_{ss}(w)=\frac{0.1+w}{ 1.1+w }$, and this function is strictly concave in~$w$.
 At steady-state, the constraint~$w\leq y$ means that $w \leq \frac{0.1+w}{ 1.1+w }$.
Thus, now the  maximization problem  is
\be\label{eq:pcvre}
\max y_{ss}(w) \text{ subject to } 0\leq w \leq \frac{0.1+w}{ 1.1+w },
\ee
 and this
constraint
 defines a convex set of admissible~$w$'s,
 so~\eqref{eq:pcvre} is a   convex optimization problem. The solution of this problem is
obtained at~$w^*=0.2702$ for which~$y_{ss}(w^*)=0.2702$.
We conclude that the optimal values correspond to~$w^*=y_{ss}^*$, and
this implies that the solution to the optimization problem~\eqref{prob:opii}
  is~$v^*=1$. Thus, we can obtain the optimal weights from
	the solution
	of the reparametrized problem.
%%%
%%%
\end{example}

The next example demonstrates a synthetic and more complex network that includes feedback connections.

\begin{example} \label{ex:complex}
%%%%%%%%%%%%%%%%%%%%%%%%%%%%%%%%%%%%%%%%%%%%%%%%%%%%%%%%%%%%%%%%%%%%%%%%%%%%%%%%%%%%%%%
	Consider the network depicted in Fig.~\ref{fig:complex_diagram}.
	The network consists of four
	RFMIOs. RFMIO~1 and RFMIO~4 have dimension~$n=4$, and
	rates~$[2, 2, 2, 2, 2]$.
	RFMIO~2 and RFMIO~3 have  dimension~$n=3$, and rates~$[1, 1, 1, 1]$.
	A unit source feeds RFMIO~1 and RFMIO~2 with proportions~$v_1$ and~$1-v_1$, respectively.
	Another control parameter,~$v_2$, determines the  division of the output of RFMIO~2.
	%%%
	The total network output is defined as~$y:=\frac{3}{4}y_3+y_4+(1-v_2)y_2$.
	  Fig.~\ref{fig:complex_diagram_res}
	  depicts the steady-state output as a function
		of the control parameters~$v_1,v_2$. It may be seen that~$y$ is
	  a concave function.
	The optimal   output value~$y^*=0.8595$ is obtained for~$v_1^*=0.48$, and~$v_2^*=0$.
	The value~$v_2^*=0$ is reasonable, as this implies that all the output~$y_2$ of RFMIO~2 goes
	directly to the network output~$y$ rather than first to RFMIO~4 and from there, indirectly,
	to~$y$.
	%%%%%%%%%%%%%%%%%%%%%%
\end{example}

\begin{figure}
	\centering
	%\scalebox{0.8}{\input{fig_ex3_diagram.eps_tex}}
	 \includegraphics[scale=0.7]{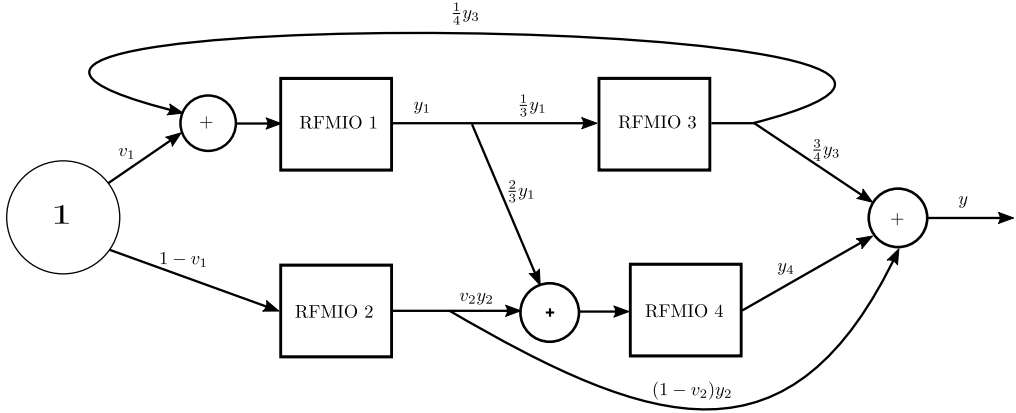}
	\caption{Topology of the network  in Example~\ref{ex:complex}.  }
	\label{fig:complex_diagram}
\end{figure}

\begin{figure}
	\centering
	%\scalebox{0.8}{\input{fig_ex3_res.eps_tex}}
	 \includegraphics[scale=0.7]{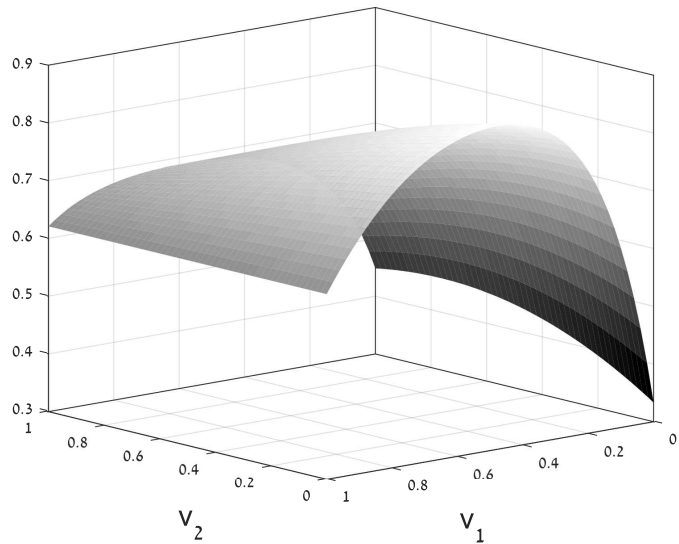}
	\caption{Steady-state output $y_{ss}(v_1,v_2)$
	for the network   in Example~\ref{ex:complex} as a function of $v_1,v_2$.}
	\label{fig:complex_diagram_res}
\end{figure}

%%%%%%%%%%%%%%%%%%%%%%%%%%%%%%%%%%%%%%%%%%%%%%%%%%%%%%%%%%%%%%%%%%%%%%%%%%%%%%%%%%%%%%%
To further explain the biological motivation of the optimization problem studied here,
consider  for example the  metabolic pathway or  protein complex described
in  Fig.~\ref{fig:operon}.
This includes a set of enzymes/proteins that are involved
 in a specific
 stoichiometry (see, for example, \cite{Lalanne2018}). In this case the objective function is of  the form~$b'y  $,
 where $y$ is a vector of production rates of the different proteins in the metabolic pathway or protein complex, and~$b$ is their stoichiometry vector.
 Fig.~\ref{fig:operon}
 depicts  an operon with four coding regions on the same transcript. The initiation rate to each ORF  is affected by an ``external'' factor (e.g. the intracellular pool of ribosomes), and also by the
 ``leakage'' of ribosomes from the previous~ORF.  The proteins produced in the operon ($P_1,\dots,P_4$) with production rates $y_1,\dots,y_4$ are part of a metabolic pathway where they are ``needed'' with a stoichiometry vector~$b = [1,3,2,1]'$.
%%%%%%%%%%%%%%%

\begin{figure}
 \begin{center}
%%%%%%%%%%%%%%%%%%%%%%%[width=12cm,height=12cm]
  \includegraphics[scale=0.7]{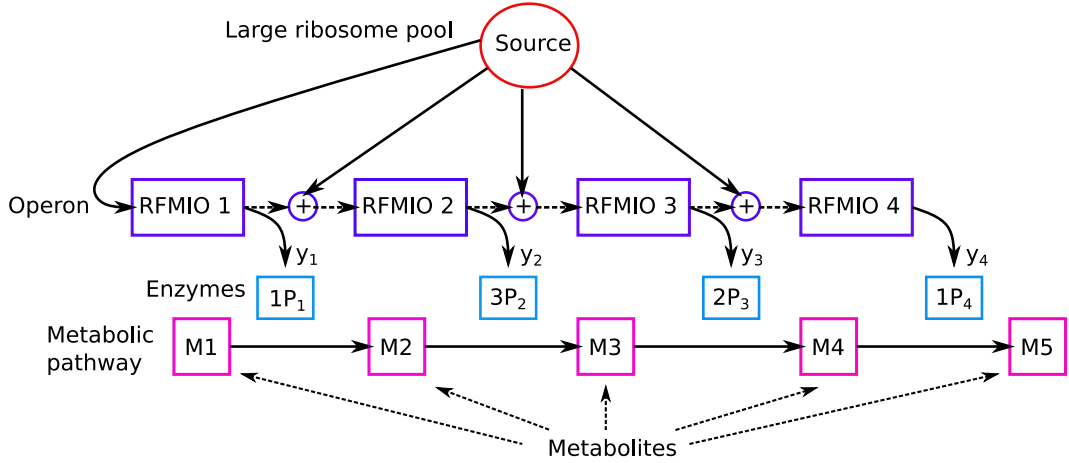}
	\caption{A   network describing
	an operon translation model and a corresponding metabolic pathway.}
	\label{fig:operon}
%%%%%%%%%%%%%%%%%%%%%%%%%%
\end{center}
\end{figure}

%%%%%%%%%%%%%%%
\section*{Analysis of feed-forward networks}\label{sec:Analysis-of-ff-networks}
%%%%%%%%%%%%%%%%%%%%%%%%%%%%%%%%%%%%%%%%%%%%%%%%%%%%%%%%%%%%
In this  section we  further  analyze \emph{feed-forward} networks of RFMIOs, where feed-forward means that for any~$j$
the input~$u_j$
 of RFMIO~$j$  does not depend either directly or indirectly on the output~$y_j$ of RFMIO~$j$.
In other words, there are no feedback connections. In terms of graph theory, this means that
the graph describing the connections is a directed acyclic graph~(DAG). \updt{For these networks
Problem~\ref{prob:maxoutgen} can \updt{be} solved in a more direct way}.

\begin{proposition}\label{prop:concave}
%%%%%%%%%%%%%%%%%%%%%%%%%%
Consider a   feed-forward network of   RFMIOs. Let $v$ denote the collection of all control weights. The mapping $v \to y_{ss}(v)$ is strictly concave.
%%%%%%%%
\end{proposition}

The following example demonstrates this result.
\begin{example} \label{ex:paralel}
%%%%%%%%%%%%
	Consider
	the network   of two RFMIOs described in Fig.~\ref{fig:ex_paralel}.
	Each   RFMIO  has dimension~$n=3$. The rates of RFMIO~1
	  are $[\lambda_0^1, \lambda_1^1, \lambda_2^1, \lambda_3^1] = [1, 1, 1, 1]$,
	and those of RFMIO~2 are~$[2, 2, 2, 2]$. In other words, every
	rate  in RFMIO~1 is slower than the corresponding rate
	in RFMIO~2.
	A unit source feeds both RFMIO~1  with $u_1=v  $,
	and RFMIO~2  with $u_2=1-v$. The network output~$y(t)$
	is defined to be the sum of the two RFMIO outputs.
	Fig.~\ref{fig:ex_paralel_res}
	depicts the network
	steady-state output  as a function of~$v\in[0,1]$.
	It may be seen that~$y_{ss}(v)$
	is a strictly concave function of~$v$, and in particular that there exists a
	unique value~$v^*=0.3971$ for which~$y_{ss}$ is maximized. This  corresponds
	to feeding a smaller [larger] part
	of the joint source to RFMIO~1 [RFMIO~2]. This is reasonable,
	as RFMIO~1 has slower rates than RFMIO~2. Hence,
	it is possible to ``direct'' more traffic
	to RFMIO~2 while still avoiding ``traffic jams'' in this RFMIO.
	%%%%%
\end{example}

\begin{figure}
	\centering
	%\scalebox{0.8}{\input{fig_ex1.eps_tex}}
	 \includegraphics[scale=0.7]{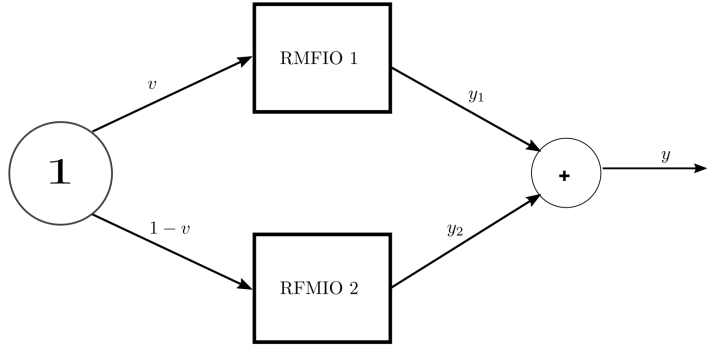}
	\caption{ Two RFMIOs fed from a common unit source.
	The input to RFMIO~1 is~$u_1(t)\equiv v \in[0,1]$, and the input to RFMIO~2 is~$u_2(t)\equiv  1-v$.
	The network output is defined as the sum  of the~RFMIO outputs~$y(t):=y_1(t)+y_2(t)$.
	\label{fig:ex_paralel}}
\end{figure}
	
\begin{figure}
 \begin{center}
\includegraphics[width= 11cm,height=8cm]{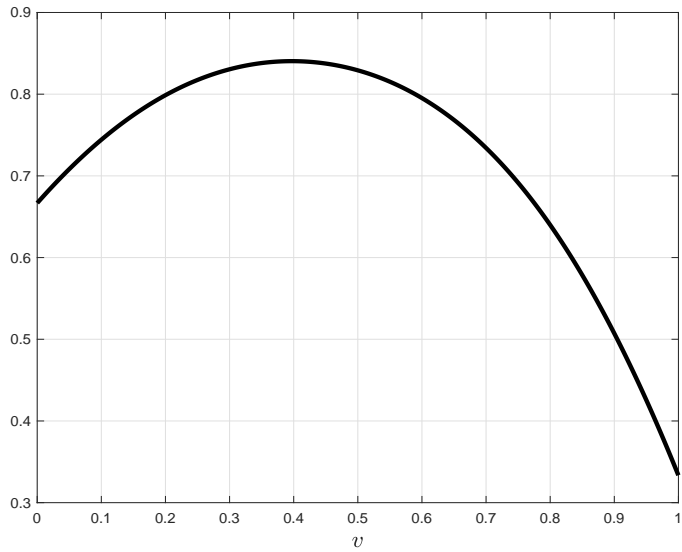}
\caption{Steady-state network output~$y_{ss}(v)$ for the network in Example~\ref{ex:paralel} as a function of the control parameter~$v$.}\label{fig:ex_paralel_res}
 \end{center}
\end{figure}

This example demonstrates the  problem of dividing
 a common resource in a biological network between
several ``clients'' such that some  overall
performance measure is optimized. For example, this network can represent the problem of optimizing
  heterologous  protein levels   by  introducing a  copy of the gene to the host genome in {\em multiple} locations. This raises the question of the optimal strength of the initiation rate of the different gene copies (engineered via   manipulation of the 5'UTR and the beginning of the~ORF, e.g. via the introduction of Shine-Dalgarno sequence with different strengths   and manipulation   of the mRNA folding in this region)~\cite{Salis2009,Shaham2017}.

%\begin{figure}[t]
%	\centering
%	\scalebox{0.64}{\input{fig_ex1_res.eps_tex}}
%	\caption{Total steady-state output~$y=y_1+y_2$  for
%	the system in Example~\ref{ex:paralel} as a function of~$v$.
%	\label{fig:ex_paralel_res}}
%\end{figure}

\subsection{Spectral representation of the network steady-state}
%%%%%%%%%%%%%%%%%%%%%%%%%%%%%%%%%%%%%%%%%%%%%%%%%%%%%%%%%%%%
Recall that in a single RFM it is possible to obtain the steady-state density (and thus the
steady-state production rate) using a spectral approach and without numerically simulating the dynamics.
The same property immediately carries over to feed-forward networks of~RFMIOs.
To explain this, consider an~RFMIO, say RFMIO~$j$,
 that is  fed only by a constant source. This is just an RFM and its steady-state density and output can be calculated as in Thm.~\ref{thm:spect}.
Now consider  an RFMIO that is fed by the output of RFMIO~$j$. Its input converges to a steady-state value~$u_{ss}$,
and since the RFMIO is contractive~\cite{RFM_entrain},   its density converges to a steady-state
that is identical to the steady-state of  an RFMIO with the constant input~$u(t)\equiv  u_{ss}$
(see e.g.~\cite{con_with_inputs,lars16}).
We can now determine the steady-state values
in the consecutive~RFMIOs  and so on.
The next example demonstrates this for a simple network.

\begin{example} \label{ex:rfms_serial_propor}
%%%%%%%%%%%%%%%%%%%%%%%%%%%%%%%%%%%%%%%%%%%%
Consider the network in Example~\ref{ex:rfms_serial} that includes $m=2$ RFMIOs both with
dimension~$n=3$.
The first RFMIO has rates $[\lambda_0, \lambda_1, \lambda_2, \lambda_3] = [1, 1, 1/4, 1]$ and input~$u(t)=1$.
The corresponding Jacobi matrix is
\[
                 \begin{bmatrix}
             %     J:= \left[ \begin{smallmatrix}
%%
 0 &  1   & 0 &0 & 0 \\
1  & 0  &  1   & 0  &0 \\
 0&  1  & 0 &  2    &  0 \\
%%%
%%%
 0& 0  &2  & 0& 1      \\
%%%
 0& 0  &0 & 1  & 0
%%%
 \end{bmatrix}.
 % \end{smallmatrix}\right].
\]
The Perron eigenvalue [eigenvector]
of this matrix is~$\sigma=\sqrt{6}$ [$\zeta=\begin{bmatrix}
1/2,\sqrt{3/2}, 5/2,\sqrt{6},1 \end{bmatrix}'$], and using Thm.~\ref{thm:spect},
 we conclude that
the steady-state density in RFMIO~1 is $e=\begin{bmatrix} 5/6&4/5&1/6\end{bmatrix}'$ (compare with Fig.~\ref{fig:ex_rfms_serial}),
 and the steady-state output is~$y_{ss}=\sigma^{-2} =1/6$.
We can now analyze the second RFMIO. Its rates are all one and the input is the output of RFMIO~1, so at steady-state the rates are~$\begin{bmatrix}  1/6&1&1&1  \end{bmatrix}'$.
The corresponding Jacobi matrix is
\[
                 \begin{bmatrix}
             %     J:= \left[ \begin{smallmatrix}
%%
 0 &  \sqrt{6}   & 0 &0 & 0 \\
 \sqrt{6}   & 0  &  1   & 0  &0 \\
 0&  1  & 0 &  1    &  0 \\
%%%
%%%
 0& 0  &1  & 0& 1      \\
%%%
 0& 0  &0 & 1  & 0
%%%
 \end{bmatrix}.
 % \end{smallmatrix}\right].
\]
The Perron eigenvalue [eigenvector]
of this matrix is~$\sigma=2.6819$ [$\zeta=\begin{bmatrix}
12.7192& 13.926& 6.1926& 2.6819& 1 \end{bmatrix}'$], and using Thm.~\ref{thm:spect},
 we conclude that
the steady-state density in RFMIO~2 is $e=\begin{bmatrix}
0.1658&0.1615& 0.139
\end{bmatrix}'$ (compare with Fig.~\ref{fig:ex_rfms_serial}),
 and the steady-state output is~$y_{ss}=\sigma^{-2} =0.139$.
%%%%%%%%%%%%
\end{example}

%%%%%%%%%%%%%%%%%%%%%%%%%%%%%%%%%%%%%%%%%%%%%%%%%%%%%%%%
\updt{Now that we have considered networks with  and without
feedback  connections, we are ready to demonstrate how
 Problem~\ref{prob:maxoutgen} can be efficiently solved.} For a feed-forward network, Prop.~\ref{prop:concave} implies  that the objective function of
Problem~\ref{prob:maxoutgen} is strictly concave. For a network with feedback connections,
Prop.~\ref{prop:concavegen} implies that Problem~\ref{prob:maxoutgen} can \updt{be} reparametrized so that it becomes strictly concave.
The first [second] constraint in~\eqref{eq:constrain} is convex [affine] and this implies the following result (see e.g.~\cite{convex_boyd}).

\begin{theorem}\label{thm:convex_opt}
%%%%%%%%%%%%%%%%%%
Problem~\ref{prob:maxoutgen} can always  be cast as  a strictly convex optimization problem.
%%%%%%%%%%
\end{theorem}

Thm.~\ref{thm:convex_opt} implies in particular that the optimal solution is unique. Moreover, there exist highly
efficient numerical algorithms for computing the unique solution even for very large networks.

%%

%%%%%%%%%%%%%%%%%%%%%%%%%%%%%%%%%%%%%%%%%%%%
\section*{A Biological Example}\label{sec:bio_exp}
%%%%%%%%%%%%%%%%%%%%%%%%%%%%%%%%%%%%%%%%%%%%
In order to demonstrate how our model can be used to address questions arising in
 synthetic biology, we consider the  problem of
  maintaining high growth rates for both a highly expressed heterologous protein and a highly expressed endogenous protein in the cell.
	These issues are currently attracting considerable interest
	as
	 lack of understanding
of the burden of expressing additional genes affects our
ability to predictively   engineer cells
(see e.g.~\cite{burden2018} and the references therein).

	Specifically, we consider the problem of maximizing the sum of
	the
	steady-state production rates of both heterologous and endogenous genes, under the assumption that they share a common ribosomal resource.
	We assume that  the coding regions of the endogenous gene (which may include various regulatory signals), and the coding region of the heterologous gene (which is optimized to include the most efficient codons) cannot be modified. However, it is possible to engineer the UTRs of these genes in order to modulate their initiation rates.
We demonstrate how this
    biological problem can be modeled and analyzed in the framework of our model.

The endogenous gene is  the highly expressed {\scrv} gene YGR192C that encodes the protein TDH3, which is involved in glycolysis and gluconeogenesis.
The heterologous gene is the green fluorescent protein~(GFP) gene (the GFP protein sequence is from gi:1543069), optimized for yeast (i.e. its codons composition was synonymously modified to consist  of optimized yeast codons).
The YGR192C gene ORF consists of $332$ codons, and the GFP gene ORF of $239$ codons.
The   simultaneous translation of these two genes while using a shared resource is modeled as
 depicted in Fig.~\ref{fig:ex_paralel},
 where RFMIO~$1$ (fed by an input $u=v$) models the translation of the~YGR192C gene, and RFMIO~$2$
 (fed by the input $u=1-v$) models the translation of the GFP gene. Here~$v\in[0,1]$ is a parameter that determines the relative amount of
ribosomal resources allocated to each gene.

Similarly to the approach used in \cite{rfm0}, we divide the YGR192C [GFP] mRNA sequnce
into $33$ [$23$] consecutive subsequences: the first  subsequence
 includes the first nine codons (that are also related to later stages of initiation~\cite{Tuller2015}). The other subsequences
  include~$10$ non-overlapping codons each,
except for the  last subsequence in the YGR192C  gene that includes~$13$ codons. This partitioning was found to
optimize the correlation between the~RFMIO predictions and biological data.

We model the translation of the YGR192C [GFP] gene using an~RFMIO with~$n=32$ [$n=23$] sites. To determine the RFMIO paramteres
 we first estimate
the  elongation rates~$\lambda_1,\dots,\lambda_n$, using ribo-seq data for the codon decoding rates~\cite{Dana2014B}, normalized so that the median elongation  rate of all {\scrv} mRNAs becomes~$6.4$ codons per second \cite{Karpinets2006}. The site rate is~$(\text{site time})^{-1}$, where site time is the sum over the decoding times of all the codons   in this site.
%%%%
These rates thus  depend on various factors including  availability of tRNA molecules, amino acids, Aminoacyl tRNA synthetase activity and concentration, and local mRNA folding~\cite{Dana2014B,Alberts2002,Tuller2015}.

The initiation rate (that corresponds to the first subsequence) for the YGR192C gene is estimated based
 on the ribosome density per mRNA levels, as this value is expected to be approximately proportional to
the initiation rate when initiation is rate limiting \cite{rfm0,HRFM_steady_state}. Again, we applied a normalization that brings the median initiation rate of all {\scrv} mRNAs to~$0.8$ mRNAs per second~\cite{Chu2014}, and this results in an initiation rate of $2.1958$ for the YGR192C gene. The GFP initiation rate was set to $0.8$.  A calculation shows that when each gene is modeled separately using an~RFMIO with $u=1$,  the steady-state production rate of the gene YGR192C [GFP]  is $r_{ss} = 0.1859$ [$r_{ss}= 0.1892$].

Fig.~\ref{fig:biol_exp}
depicts the network output $y_{ss}(v)$, as a function of $v\in[0,1]$. The unique maximum $y_{ss}(v^*)=0.3429$ is attained for $v^*=0.4311$, which corresponds to feeding a smaller [larger] part of the common ribosomal
resource to the GFP [YGR192C] gene. This is reasonable, as the steady-state production rate of the GFP gene is slightly larger than the steady-state production rate of the YGR192C gene. This result implies that in order to maximize the sum of the steady-state production rates of the YGR192C gene and the GFP gene, using a common ribosomal resource, their UTRs binding efficiency should be engineered such that~$43\%$ of the ribosomal resource initiates the YGR192C mRNAs, and the remaining~$57\%$ initiates the GFP mRNAs. Our analytical approach can also be used   to
determine the ribosomal allocation that maximizes
 some weighted sum of the two  production rates.

\begin{figure}
 \begin{center}
\includegraphics[width= 9cm,height=6cm]{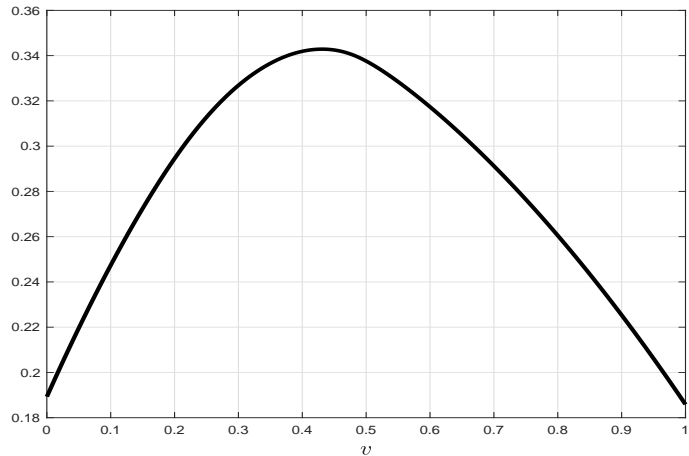}
\caption{Steady-state network output~$y_{ss}(v)$ for the of simultaneous translation of YGR192C (RFMIO $1$) and GFP (RFMIO $2$) genes, using the network depicted in Fig~\ref{fig:ex_paralel}, as a function of the control parameter~$v$.}\label{fig:biol_exp}
 \end{center}
\end{figure}

%%%%%%%%%%%%%%%%%%%%%%%%%%%%%%%%%%%%%%%%%%%%
\section*{Discussion}\label{sec:discussion}
%%%%%%%%%%%%%%%%%%%%%%%%%%%%%%%%%%%%%%%%%%%%

Studying  the flow of biological
``machines'' like ribosomes, RNAPs, or motor proteins along
biological networks like interconnected mRNA molecules or filaments is of paramount importance.
These biological
machines have volume and thus  satisfy a simple  exclusion principle: two machines
 cannot be in  the same place at the same time.

In order to better understand these cellular  biological processes it is important to study the flow of biological
``machines'' along networks of interconnected ``roads'', and not only in isolated processes.
We propose to model and analyze such phenomena  using networks of interconnected~RFMIOs.
The RFMIO dynamics
satisfies a  ``soft''  simple exclusion principle: as a site becomes fuller the effective entry rate into this site decreases. In particular, ``traffic jams'' may evolve behind slowly moving machines.
The input and output of every RFMIO facilitate their
integration into interconnected networks that can be represented using a graph.  The  nodes in this graph
represent the different RFMIOs,
 and the (weighted) edges describe  how the output of each RFMIO is divided between  the inputs
of other (or the same) RFMIO. Our main result shows that under quite general positive feedback connections, such a network always admits a unique steady-state
and any trajectory of the system converges to this steady-state. This opens the door to many interesting research questions, e.g., how does this  steady-state depends
on the various parameters in the network, and for what  feedback connections
is the steady-state output of the network optimized?

We demonstrated using various examples how such networks can model
 interesting biological phenomena like  competition for shared resources, the optimal distribution
of a shared biological resource between several ``clients'',
optimizing the effect of ribosome recycling,  and more.

The RFM is amenable to rigorous analysis, even when the rates are  not
 homogeneous,  using various tools from systems and control theory like
contraction theory~\cite{cast_book,Arcak20111219}, the theory of cooperative dynamical systems~\cite{Ji-Fa01091994,alexander2017,Raveh2016,ent__mast_eq}, convex analysis and more.
This  amenability to rigorous analysis
 carries  over to networks of~RFMIOs. For example,
     the problem of how to connect
  the  RFMIO outputs in the network to  inputs so that the steady-state
network output is maximized
can be cast as  a convex optimization problem. This means that the problem can be
solved efficiently even for large networks~(see, e.g.~\cite{convex_boyd}).

The networks  we propose here   allow modeling complex biological processes
in a coherent and useful manner.
The network models  static connections between the RFMIOs.
The dynamical part is described by the set of~ODEs for each RFMIO.
The parameters used in these models can be inferred based on various sources of large-scale genomic data (see, for example,~\cite{Cohen2018,Dana2014B}) and/or can be predicted directly from the nucleotide sequence of the gene (see, for example,~\cite{Shaham2017,Sabi2017}).
In addition, the analyzed network can be built gradually, one module after another. For example, one  can engineer and study one metabolic pathway  and then   connect
 another pathway to the existing module  (see Fig.~\ref{fig:operon}),
 etc. This yields
 a combined model that describes both biophysical aspects of  gene expression regulation (e.g. translation), and properties of metabolism (e.g. stoichiometry of enzymes and metabolites, and rates of metabolic  reactions).
%The analysis of such model may includes tools for metabolic networks analysis such as Flax Balance Analysis (FBA) \cite{Kauffman2003}.

Topics for further research include
networks  where the weighted connections between the~RFMIOs may also change with time.
This may model for example mRNA molecules that diffuse through  the cell and consequently
change their interactions  with ribosomes,  other mRNAs, etc.
(see. e.g.~\cite{Korkmazhan2017}).

Finally, networks of interconnected~TASEPs have been used to model other natural and artificial phenomena such as vehicular traffic and evacuation dynamics. We believe that the deterministic networks
 proposed  here can also be applied to model and analyze such  phenomena.

%%%%%%%%%%%%%%%%%%%%%%%%%%%%%%%%%%%%
\section*{Appendix: Proofs}
%%%%%%%%%%%%%%%%%%%%%%%%%%%%%%%%%%%%

%%%%%%%%%%%%%%%%%%%%%%%%%%%%%%
The proof of Thm.~\ref{thm:ss}
 is based on showing that the network of~RFMIOs is a cooperative dynamical system~\cite{hlsmith} whose trajectories evolve on a compact state-space
and with  a \emph{unique} equilibrium point\footnote{In this Appendix we use  the term equilibrium point instead of
 steady-state.}   in this state-space. We require the following auxiliary result.
 %For two vectors~$p,q\in \R^n$, we write~$p\ll q$ if~$p_i<q_i$ for all~$i$.

\begin{theorem}\label{thm:unique_ex}
Consider a network of~$m$ RFMIOs  in the form:
%%%%%%
\begin{align}\label{eq:netisovv}
								\dot x^1 &=f^1(x^1,u^1),   &  y^1&=\lambda^1_{n^1} x^1_{n^1} , \nonumber \\
								\vdots\\
								\dot x^m &=f^1(x^m,u^m),   & y^m&=\lambda^m_{n^m} x^m_{n^m}  , \nonumber
\end{align}
with the inputs   given by
\be \label{eq:feedbacko}
				u^i=c^i_0+ \sum_{k=1}^m c_k^i y^k,\quad i=1,\dots,m ,
\ee
where~$c^i_0>0$, and~$c_k ^i\geq  0$ for~$k=1,\dots,m$.
This network admits no  more than a
 single equilibrium point in the 
state-space~$ (0,1)^{n^1}\times\dots\times (0,1)^{n^m}
$.
 %%%%%%%%%%%%
\end{theorem}

Note that this represents a quite general network   as the input to every RFMIO may include
a   contribution from the output of \emph{every} RFMIO in the network, with nonnegative weights.
The technical condition~$c^i_0>0$ is needed to guarantee that~$u^i(t)>0$ for all~$t$ and, in particular,~$u_{ss}^i>0$.

%%%%%%%%%%%%%%%%%%%%%%%%%%%%%%%%%%%%%%%%%%%%%%%%%%%%%%
\begin{proof}[Proof of Theorem~\ref{thm:unique_ex}]
%%%%%%%%%%%%%%%%%%%%%%%%%%%%%%%%%%%%%%%%%%%%%%%%%%%%%%%
We begin by considering the case~$m=1$.
In this case, the dynamics is
\begin{align}\label{eq:dvdfs}
								\dot x_1&=\lambda_0 ( c_0+ c_1  \lambda_n x_n ) (1-x_1)-\lambda_1 x_1(1-x_2), \nonumber \\
								\dot x_2&=\lambda_1 x_1(1-x_2)-\lambda_2 x_2(1-x_3), \nonumber\\
								   &\vdots \\
									\dot x_n&=\lambda_{n-1} x_{n-1}(1-x_n)-\lambda_n x_n, \nonumber
\end{align}
with~$c_0>0$, and~$c_1\geq 0$.
Suppose that~$e \in (0,1)^n$ is an   equilibrium point.
Eq.~\eqref{eq:dvdfs} yields
\begin{align} \label{eq:clsdp}
								 \lambda_0 ( c_0+ c_1  \lambda_n e_n ) (1-e_1)&=\lambda_1 e_1(1-e_2) \nonumber \\
								                          &= \lambda_2 e_2(1-e_3)\nonumber\\
								                          &\vdots \\
									                        &=\lambda_{n-1} e_{n-1}(1-e_n) \nonumber\\
																					 &=\lambda_n e_n. \nonumber
\end{align}
It follows
that~$e_n$ uniquely determines~$e_{n-1}$. Then~$e_n,e_{n-1}$ uniquely determine~$e_{n-2}$, and so on.
 We  conclude that~$e_n$ uniquely determines~$e$.
Suppose that~$\tilde e $, with~$\tilde e \not = e$, is another  equilibrium point in~$(0,1)^n$.
Then~$\tilde e_n \not = e_n$, and we may assume that
\be \label{eq:enntil}
\tilde e_n<e_n.
\ee
Eq.~\eqref{eq:dvdfs} yields
\[
					  \lambda_{n-1} \tilde e_{n-1}(1-\tilde e_n ) 	=\lambda_n \tilde e_n  <
						\lambda_n   e_n
					  =\lambda_{n-1} e_{n-1}(1- e_n ) ,
\]
so
\[
\tilde e_{n-1}<e_{n-1}.
\]
Continuing in this fashion yields
\be\label{eq:podscdx}
 \tilde e_i<e_i , \quad i=1,\dots,n.
\ee
On the other-hand,~\eqref{eq:clsdp} yields
\begin{align*}
											 e_1-\tilde e_1&=\frac {\lambda_n \tilde e_n}{\lambda_0 ( c_0+ c_1  \lambda_n \tilde e_n )}-\frac {\lambda_n e_n}{\lambda_0 (c_0+ c_1  \lambda_n e_n )}\\
											&=\frac{\lambda_n c_0(\tilde e_n-e_n)}{\lambda_0  ( c_0+ c_1  \lambda_n \tilde e_n )  ( c_0+ c_1  \lambda_n e_n )}
											\\&< 0.
\end{align*}
This contradicts~\eqref{eq:podscdx}, so we
 conclude that when~$m=1$ the  network
 admits no more than a single equilibrium.

We now consider the case~$m>1$.
%%%%%%%%%%%%%%%
Suppose that~$e\in (0,1)^{n^1} \times \dots\times (0,1)^{n^m}$ is an equilibrium point.
Write~$e=\begin{bmatrix} e^1 \\\vdots\\ e^m\end{bmatrix}$, where~$e^i :=\begin{bmatrix} e^i_1 &\dots&e^i_{n^i}    \end{bmatrix}'$.
 For~$i=1,\dots,m$, let~$r^i(e):=  c_0^i+ \sum_{k=1}^m c_k^i  \lambda^k_{n^k} e^k_{n^k}   $, i.e. the
steady-state input to   the~$i$th RFMIO. Then
at steady-state
\begin{align} \label{eq:newq}
								 \lambda_0 ^i r^i(e) (1-e_1^i)&=\lambda_1^i e_1^i(1-e_2^i) \nonumber \\
								                          &= \lambda_2^i e_2^i(1-e_3^i)\nonumber\\
								                          &\vdots \\
									                        &=\lambda_{{n^i}-1}^i e_{{n^i}-1}^i(1-e_{n^i}^i) \nonumber\\
																					 &=\lambda_{n^i}^i e_{n^i}^i. \nonumber
\end{align}
We already know  that~$e^i_{n^i}$ uniquely determines~$e^i$.
 Suppose that~$\tilde e\not = e$ is another equilibrium point of the network.
Then~$ \tilde e^i_{n^i} \not = e^i_{n^i}$ for some~$i$.
We may assume that~$\tilde  e^1_{n^1}   < e^1_{n^1}$.
 Arguing  as in the case~$m=1$ above yields
\be\label{eq:atruhn}
 \tilde e_i^1<e_i^1 , \quad i=1,\dots,n^1.
\ee
On the other-hand,~\eqref{eq:newq} yields
\begin{align*}
											 e_1^1-\tilde e_1^1&=\frac {\lambda_{n^1}^1 \tilde e_{n^1}^1}{\lambda_0^1  r^1(\tilde e)}-\frac {\lambda_{n^1}^1 e_{n^1}^1}{\lambda_0^1 r^1(e)}\\
											&=   \lambda_{n^1}^1
											\frac{   c_0^1 ( \tilde e_{n^1}^1- e_{n^1}^1 ) +\sum_{k=2}^m
													c^1_k \lambda^k_{n^k} ( \tilde e_{n^1}^1 e^k_{n^k} -   e_{n^1}^1 \tilde e^k_{n^k}  ) }{\lambda_0^1 r^1(\tilde e)r^1(  e)} .
\end{align*}
Combining this with~\eqref{eq:atruhn}
  implies that at least one of the terms in the summation on the right-hand side  must be  positive.
We may assume that
\be\label{eq:omuidt}
\tilde e_{n^1}^1 e^2_{n^2} >   e_{n^1}^1 \tilde e^2_{n^2}.
\ee
Thus, $ \frac{ e^2_{n^2} }{\tilde e^2_{n^2}  } > \frac{  e_{n^1}^1}{\tilde e_{n^1}^1}>1$,
and  we conclude that
\be\label{eq:atruhn2}
 \tilde e_i^2<e_i^2 , \quad i=1,\dots,n^2.
\ee
Now~\eqref{eq:newq} yields
\begin{align*}
											  e_1^2-\tilde e_1^2&=\frac {\lambda_{n^2}^2 \tilde e_{n^2}^2}{ \lambda_0^2 r^2(\tilde e)}-\frac {\lambda_{n^2}^2 e_{n^2}^2}{\lambda_0^2 r^2(e)}\\
											&= \lambda_{n^2}^2
											\frac{   c_0^2 ( \tilde e_{n^2}^2- e_{n^2}^2 ) +c^2_1 \lambda^1_{n^1} ( \tilde e_{n^2}^2 e^1_{n^1} -   e_{n^2}^2 \tilde e^1_{n^1}  )+\sum_{ k =2}^m
													c^2_k \lambda^k_{n^k} ( \tilde e_{n^2}^2 e^k_{n^k} -   e_{n^2}^2 \tilde e^k_{n^k}  ) }{\lambda_0^2 r^2(\tilde e)r^2(  e)} .
\end{align*}
Combining  this with~\eqref{eq:omuidt}
and~\eqref{eq:atruhn2} implies that at least one of the terms in the summation  must be  positive. We may assume that
\[
 \tilde e_{n^2}^2 e^3_{n^3} >   e_{n^2}^2 \tilde e^3_{n^3} ,
\]
i.e.
\[
									\frac{ e^3_{n^3} }{\tilde e^3_{n^3}  } > \frac{  e_{n^2}^2}{\tilde e_{n^2}^2}>1.
\]
Continuing in this manner, we find that
\be\label{eq:efty}
%%%
							\frac{ e^m_{n^m} }{\tilde e^m_{n^m}  } >\dots > \frac{  e_{n^2}^2}{\tilde e_{n^2}^2}>  \frac{  e_{n^1}^1}{\tilde e_{n^1}^1}>1	,
%%%
\ee
and that
\be\label{eq:poutt}
\tilde e^k_j < e ^k_j , \quad k=1,\dots,m,\; j=1,\dots,n_k.
\ee
Using~\eqref{eq:newq} again yields
\begin{align*}
											  e_1^m-\tilde e_1^m
												&= \lambda_{n^m}^m
											\frac{   c_0^m ( \tilde e_{n^m}^m- e_{n^m}^m )  +\sum_{ k =1}^{m-1}
													c^m_k \lambda^k_{n^k} ( \tilde e_{n^m}^m e^k_{n^k} -   e_{n^m}^m \tilde e^k_{n^k}  ) }{\lambda_0^m r^m(\tilde e)r^m(  e)} .
\end{align*}
By~\eqref{eq:efty} and \eqref{eq:poutt},
 the left-hand side here is positive and the right-hand side is negative.
This contradiction completes the proof.
\end{proof}

We can now prove  Thm.~\ref{thm:ss}. For a set~$\W$ we denote the interior of~$\W$ by~$\inter(\W)$.
%%%%%%%%%%%%%%%%
\begin{proof}[Proof of Thm.~\ref{thm:ss}]
%%%%%%%%%%%%%%%%%%%%%%%%%%%%%%%
Write the network~\eqref{eq:netisovv} and~\eqref{eq:feedbacko}
 as~$\dot x= f(x)$, with~$x\in\R^{n^1+\dots+n^m}$.
Let~$J(x):=\frac{\partial}{\partial x}f(x)$ denote the Jacobian of this dynamics.
We claim that~$J(x)$ is Metzler for all~$x\in \Theta:=[0,1]^{n^1}\times\dots\times [0,1]^{n^m}$.
 Indeed, every RFMIO is a cooperative system, so it has a Metzler Jacobian, and  the other non-zero off-diagonal  terms in~$J$ are due to the connections~\eqref{eq:feedbacko}  and are nonnegative, as all the~$c^i_k$s are nonnegative. We conclude that the network is a cooperative system. It is not difficult to show that~$\Theta$
 is an invariant set, and since it is convex and compact it admits at least one equilibrium point~$e$.
Furthermore, it can be shown that for any initial   condition~$a\in\Theta$ the  solution satisfies~$x(t,a)\in \inter(\Theta)$ for all~$t>0$.
Thus, any equilibrium point satisfies~$e\in\inter(\Theta)$.
 By Thm.~\ref{thm:unique_ex} the network
  admits a unique equilibrium point~$e$. Now Ji-Fa's Theorem~\cite{Ji-Fa01091994} implies that~$e$ is~GAS.
 %%%%%%%%%%%%%%%%%
\end{proof}

%%%%%%%%%%%%%%%%%%%%%%%%%%%%%%%%%%%%%%%%%%
\begin{proof}[\bf Proof of Prop.~\ref{prop:concavegen}]
%%%%
For the sake of simplicity, we detail the proof for the of network with~$m=2$ RFMIOs (the proof in the general case is very similar). In this case, the inputs to the two RFMIOs are
\begin{align*}
				u^1(t)&=c^1_0 y^0 + c_1^1 y^1(t)+c_2^1 y^2(t) ,\\
				u^2(t)&=c^2_0 y^0 + c_1^2 y^1(t)+c_2^2 y^2(t) ,
\end{align*}
where~$y^0\geq 0$  represents a constant source (e.g. a pool of ribosomes),
$c_k ^1,c_k^2\geq  0$ for~$k=0,1,2$, and
  $c_i^1+c_i^2\leq 1$ for~$i=0,1,2$.
Let
\be\label{eq:defwis}
w^k_i:=c^k_i y^i. \ee
  Then
\begin{align}\label{eq:lkiyyr}
				u^1&=w^1_0+ w^1_1 + w^1_2, \nonumber \\
				u^2&=w^2_0+  w^2_1 + w^2_2 ,
\end{align}
and the constraints become
\begin{align}\label{eq:ljneew}
								w^k_j \geq  0, w^1_i+w^2_i\leq y^i.
\end{align}
We   know that the network of RFMIOs converges to a steady-state, and that the steady-state
output~$y_{ss}^i$, $i=1,2$, is a strictly concave function of the rates in RFMIO~$i$.
In particular,~$y_{ss}^i=p_i(u^i_{ss})$, for some   strictly concave function~$p_i$.
Thus, the steady-state network output is
\begin{align}\label{eq:Lyddpo}
y_{ss}&=\sum_{j=1}^{m} v_{j,m+1} y^j_{ss}\nonumber \\
      &=\sum_{j=1}^{m} v_{j,m+1} p_j(   (w^j_0+ w^j_1 + w^j_2  )_{ss} )
%%%
\end{align}
This shows that~$y_{ss}$ is strictly concave in the steady-state~$w^i_k$s. At steady-state,
the constraints~\eqref{eq:ljneew} become
%%%%%%%
\begin{align}\label{eq:ljneewafter}
								(w^k_j)  \geq  0, w^1_i+w^2_i\leq  p_i(w^i_0+ w^i_1 + w^i_2)  .
\end{align}
  This constraint defines a convex set of admissible~$w^i_k$s. We
conclude that the problem of maximizing~\eqref{eq:Lyddpo}
 subject to~\eqref{eq:ljneewafter} is a convex optimization problem.
Determining  the optimal~$w^i_k$s is thus numerically tractable even for large networks.
Once these values are known, we can compute: (1)~the optimal steady-state~$u^i$s from~\eqref{eq:lkiyyr};
(2)~the optimal steady-state outputs~$y^i$ (e.g. using the spectral representation);
and finally (3)~the optimal weights~$c_i^k$
from~\eqref{eq:defwis}.
%%%%
\end{proof}

\begin{proof}[\bf Proof of Prop.~\ref{prop:concave}]
%%%%%%%%%%%%%%%%%%%%%%%%%%%%%%%%%
In a feed-forward network with~$m$ RFMIOs, the RFMIOs can be 
divided  into~$w$   disjoint sets  in the following manner.
Let $\mathcal O_1 \subset \{1,\dots,m\}$ denote the subset of RFMIOs 
 that are fed only from   constant sources.
Similarly, let~$\mathcal O_2 \subset \{1,\dots,m\}\setminus  \mathcal O_1$ denote the subset
of RFMIOs that are fed   from the outputs of 
RFMIOs in~$\mathcal O_1$ and/or from constant sources,
  but  that are not in~$\mathcal O_1$,  and so on.
	Note that for~$i\not =j$, $\mathcal O_i\cap \mathcal O_j=\emptyset$, and that~$ \mathcal O_1\cup \cdots\cup\mathcal O_w=\{1,\dots,m\}$.
It has been shown in~\cite{rfm_max} that in the RFM
the mapping~$(\LMD) \to r_{ss}$  is  strictly concave   over~$\R^{n+1}_{++}$. In particular, the mapping from $\lambda_0$ to $r_{ss}$ is strictly concave. This implies that in an RFMIO
with a positive constant input~$u(t)\equiv v$  the mapping $v \to y_{ss}$ is strictly concave.
Consider  a feed-forward network of RFMIOs.
Pick an RFMIO in~$\mathcal O_1$. The input to this RFMIO has the form~$v^1 u^1+\dots +v^p u^p$,
where the~$u^i$s are positive sources and the~$v^i$s are control weights. It follows that
the mapping~$(v^1,\dots,v^p)\to y_{ss}$   of this RFMIO is strictly concave.
We conclude that any weighted sum of outputs of RFMIOs in~$\mathcal O_1$
 is strictly
 concave in the (relevant) control weights. We can now proceed  to RFMIOs in~$\mathcal O_2$,
 and so on.
\end{proof}

 %Ref.~\cite{rfm_max} showed that the steady-state production rate in a RFM is a strictly concave function of the parameters $\lambda_i$, $i=0,\dots,n$. This implies that all the functions $R_\alpha(u)$ are concave, which means that the performance index in~\eqref{eq:p_index} in our control problem is strictly concave. Since  the constraints from~\eqref{eq:constrain} are convex, the control vector $v=(v_{\alpha\beta})$ is contained within a convex set.

%Thm,~\ref{THM:1} is a corollary of a result from \cite{mt}, which says that the output $y(\lambda)$ of a single RFMIO with parameters $\lambda=(\lambda_0,\dots,\lambda_n)$ is a strictly {\em concave} function of the vector parameter $\lambda$.  This fact implies immediately that all functions $R_\alpha(u)$ are concave. This, in turn, yields that the performance performance index in~\eqref{eq:p_index} in our control problem is concave, while the constraints from~\eqref{eq:constrain} are convex: the control vector $v=(v_{\alpha\beta})$ is contained within a convex set.

\section*{Acknowledgments}
%%%%%%%%%%%%%%%%%%%%%%%%%%%%%%%%%
YZ gratefully acknowledges the support of the Israeli Ministry of Science, Technology and Space.
 This work was partially supported by a grant from the
Ela Kodesz Institute for Cardiac Physical Sciences and Engineering.
The work of~MM is supported by a research grant from the Israeli Science Foundation~(ISF).
The work of~MM and~TT is supported by a research  grant from the US-Israel Bi-national Science foundation~(BSF).
%%%
The work of AO was conducted in the framework of the state project no. AAAA-A17-117021310387-0 and is partially supported by RFBR grant 17-08-00742.

%%%%%%%%%%%%%%%%%%%%%%%%%%%%%%%%%%%%%%%%%%%%%%%%%%

\section*{References}
%%%%%%%%
\bibliographystyle{IEEEtranS}
%\bibliographystyle{nature}  % alphabetically ordered references
%\bibliography{rfm,RFM_bibl,rfm_nets}

%%%%%%%%%%%%%%%%%%%%%%%%%%%%%%%%%%%%%

\section*{Author contributions}
I.O., Y.Z., A.O., T.T. and M.M. performed the research and wrote the paper.

\section*{Additional Information}
\subsection*{Competing financial interests.}
The authors declare no competing financial \updt{or non-financial} interests.

%\section*{Figures}
%%%%%%%%%%%%%%%%%%%%%%%%%%%%%%%%%%%%%%%%%%%%%%%%%%%%%%%

\end{document}